\documentclass[10pt,twocolumn,twoside]{IEEEtran} 
\ifCLASSINFOpdf
\else
\fi

\usepackage{amsmath,amsfonts,amssymb,color}
\usepackage[dvipsnames]{xcolor}
\definecolor{winered}{rgb}{0.8,0,0}
\definecolor{myblue}{rgb}{0,0,0.8}
\usepackage{tikz}
\usepackage{amsthm}
\usepackage{empheq}

\newtheorem{problem}{Problem}

\newtheorem{definition}{Definition}
\newtheorem{theorem}{Theorem}
\newtheorem{lemma}{Lemma}

\newtheorem{corollary}{Corollary}
\newtheorem{remark}{Remark}
\newtheorem{assumption}{Assumption}
\newtheorem{example}{Example}
\usepackage{epsfig} 
\usepackage{subcaption}

\DeclareMathOperator*{\argmax}{\arg\!\max}

\usepackage{algorithm}
\usepackage{algpseudocode}

\usepackage{epsfig}
\usepackage{array}
\usepackage{multirow}
\usepackage{epstopdf}
\usepackage{tikz}
\usepackage{relsize}
\usetikzlibrary{shapes,arrows}
\usepackage[hidelinks]{hyperref}
\hypersetup{
    colorlinks=true,
    linkcolor={winered},
    citecolor={myblue}
}
\usepackage{cite}

\begin{document}
%
\title{On the Impacts of Redundancy, Diversity,  and Trust in Resilient Distributed State Estimation} 

\author{Aritra~Mitra, Faiq~Ghawash, Shreyas~Sundaram and Waseem~Abbas
\thanks{This work was supported by NSF CAREER award
1653648.}
\thanks{A. Mitra and S. Sundaram are with the School of Electrical and Computer Engineering, Purdue University, West Lafayette, IN 47907 USA. Email: {\tt \{mitra14, sundara2\}@purdue.edu}}
\thanks{F. Ghawash is with the Department of Electrical Engineering, Information Technology University, Lahore, Pakistan. Email: {\tt faiq.ghawash@itu.edu.pk}}
\thanks{W. Abbas is with the Department of Electrical Engineering and Computer Science, Vanderbilt University, Nashville, TN 37212. Email: {\tt waseem.abbas@vanderbilt.edu}}
}

\maketitle
\IEEEpeerreviewmaketitle
\begin{abstract}
     We address the problem of distributed state estimation of a linear
dynamical process in an attack-prone environment. Recent attempts to
solve this problem impose stringent redundancy requirements on the
measurement and communication resources of the network. In this paper,
we take a step towards alleviating such strict requirements by
exploring two complementary directions: (i) making a small subset of
the nodes immune to attacks, or ``trusted'', and (ii) incorporating
diversity into the network. We define graph-theoretic
constructs that formally capture the notions of redundancy, diversity,
and trust. Based on these constructs, we develop a resilient estimation
algorithm and demonstrate that even relatively
sparse networks that either exhibit node-diversity, or contain a small
subset of trusted nodes, can be just as resilient to adversarial
attacks as more dense networks. Finally, given a finite budget for
network design, we focus on characterizing the complexity of (i)
selecting a set of trusted nodes, and (ii) allocating diversity, so as
to achieve a desired level of robustness. We establish that,
unfortunately, each of these problems is NP-complete.
\end{abstract}
\section{Introduction}
The distributed state estimation problem, in its most basic form, concerns asymptotic reconstruction of the state of a dynamical process, via a group of sensor nodes interacting over a network \cite{dist3,ugrinov,kim,martins,mitraTAC,wang,han,rego,nozal}. Each node observes only a portion of the state dynamics and, hence, is reliant on local information exchanges with neighboring nodes for tracking the entire state. An underlying assumption that runs through almost all works on this topic is that the sensor nodes work \textit{collaboratively} towards the common goal of state estimation. However, the recent surge of activity devoted to the security of networked control systems suggests that this may no longer be a reasonable assumption to make. Thus, it is of prime importance to design algorithms and networks that are robust to attacks on certain parts of the system. For the specific problem under consideration, there are only a few existing methods that have attempted to address this concern. These works can be broadly classified in terms of the assumptions made on the adversarial model. For example, while  \cite{deghat,junsoo,mustafa,he} consider attack models that are limited in scope, \cite{mitra_auto,mitraAR} account for worst-case Byzantine adversarial attacks \cite{Byz}. However, allowing for sophisticated attack models comes at the expense of rather stringent requirements on the communication network topology. Specifically, the guarantees provided in \cite{mitra_auto,mitraAR} hold only when the network exhibits a sufficient amount of redundancy in both its measurement and communication resources. We are thus motivated to ask: \textit{Can one relax the redundancy requirements on the network, and yet, tolerate a worst-case attack model?} The goal of this paper is to demonstrate that this can indeed be done.

Recently, in \cite{abbas} and \cite{faiq}, two distinct ideas were proposed that depart from the conventional approach of increasing robustness through redundancy. In \cite{abbas}, the authors explored the concept of device hardening, wherein a small subset of carefully selected nodes, called \textit{trusted nodes}, were made immune to attacks. On the other hand, in \cite{faiq}, the authors exploited the fact that the components of a large-scale networked control system are typically quite \textit{diverse} in their hardware and software implementations. Such diversity, in turn, implies that the vulnerabilities of different components are not necessarily alike. The key observation here is that even if an adversary manages to breach the security of a particular type of component, its impact would remain limited to only components of that type. In the context of consensus, when the above ideas are leveraged appropriately, it has been shown that even a relatively sparse network with trusted nodes \cite{abbas}, or sufficient diversity \cite{faiq}, can still  exhibit the same functional robustness as that of a highly connected, dense network.

In light of the above developments, it is natural to ask whether the ideas of \textit{trust} and \textit{diversity} can be adapted to solve the resilient distributed state estimation problem. We note that the problem at hand differs on several counts from the typical consensus setting. Indeed, the former entails tracking the state of an external (potentially unstable) dynamical system using sensor nodes that are heterogeneous in terms of their observations, features that are not exhibited by the basic consensus problem. Consequently, while we borrow ideas from \cite{abbas} and \cite{faiq}, our techniques differ considerably from these works. The main questions of interest to us are as follows.

\begin{itemize} 
\item Can introducing trusted nodes and diversity into a sparse network alleviate the redundancy requirements needed for resilient distributed state estimation?

\item How should one choose a set of trusted nodes, and incorporate diversity, such that the resulting network is endowed with a desired level of robustness? 
\end{itemize}

In posing the above questions, our primary  motivation is to gain insights regarding the design of an attack-resilient, robust sensor network. The multitude of applications of such sensor networks, and the growing need for designing secure networked control systems, justifies the relevance of the questions asked in this paper. In this context, our main contributions are summarized as follows. 

\textbf{Contributions}: In Section \ref{sec:floc_graph_algo}, we introduce novel graph-theoretic constructs that formally capture the three facets of interest, namely redundancy, diversity,  and trust. Intuition dictates that the lack of any one of these facets should be compensated by the presence of at least one of the other two - this is an intrinsic feature of the topological properties we introduce. We then develop an attack-resilient, provably-correct filtering algorithm that exploits redundancy, diversity, and trust to enable each non-compromised node to asymptotically recover the entire state, provided the graph-theoretic conditions introduced in Section \ref{sec:floc_graph_algo} are met. 

One of the assumptions typically made while dealing with Byzantine attack models is that the number of compromised nodes is bounded in some appropriate sense \cite{broad,vaidyacons,rescons,dibaji,abbas,faiq,usevitch1,mitra_auto,mitraAR,Sundaramopt,su,Byz,mitra2018impact} - an assumption that we relax in Section \ref{sec:mono_chrom_adv}. In particular, once an adversary has managed to breach the security of a particular type of component (node), we allow it to compromise any number of nodes of that type. We show how one can account for such scenarios as long as the network is sufficiently diverse in its measurement and communication resources. In the process, we argue (see Remark \ref{rem:spoofing}) that one can employ diversity as a means to tackle spoofing attacks, where an attacker can impersonate the identities of multiple nodes. 

Finally, we turn to the problem of designing a robust network subject to cost constraints. Given a certain budget that caps the number of nodes that can be made trusted, or the amount of diversity that can be afforded, we focus on understanding (i) which nodes should be made trusted, and (ii) how one should allocate diversity, in order to achieve a desired level of robustness. In Section \ref{sec:design}, we formulate these problems as decision problems and characterize their complexity. We show that, unfortunately, each of these problems is NP-complete.

A preliminary version of this paper appeared as \cite{mitra2018impact}, where we only considered the impact of making certain nodes trusted. 
\section{Notation, Terminology, and Problem Setup}
\label{sec:model}
In this section, we formally describe the various models considered throughout the paper; subsequently, we state the problem of interest. We begin by introducing relevant notation.

\textbf{Notation:} A directed graph is denoted by $\mathcal{G} =(\mathcal{V},\mathcal{E})$, where $\mathcal{V} =\{1, \cdots, N\}$ is the set of nodes and $\mathcal{E} \subseteq \mathcal{V} \times \mathcal{V} $ represents the edges. An edge from node $j$ to node $i$, denoted by (${j,i}$), implies that node $j$ can transmit information to node $i$. The neighborhood (or in-neighborhood) of the $i$-th node is defined as $\mathcal{N}_i \triangleq \{j\,|\,(j,i) \in \mathcal{E} \}.$ A node $j$ is said to be an out-neighbor of node $i$ if $(i,j)\in\mathcal{E}$. The notation $|\mathcal{V}|$ is used to denote the cardinality of a set $\mathcal{V}$. The set of all eigenvalues (or modes) of a matrix $\mathbf{A}$ is denoted by $sp(\mathbf{A}) = \{\lambda \in \mathbb{C}\,|\,det(\mathbf{A}-\lambda\mathbf{I}) = 0\}$, and the set of all unstable eigenvalues by $\Lambda_{U}(\mathbf{A}) = \{\lambda \in sp(\mathbf{A})\,|\, |\lambda| \geq 1 \}$. The identity matrix of dimension $r$ is denoted $\mathbf{I}_r$, and $\mathbb{N}_{+}$ is used to refer to the set of all positive integers. The terms `communication graph' and `network' are used interchangeably.

\textbf{Plant and Observation Model:} Consider a linear time-invariant dynamical process
\begin{equation}
\mathbf{x}[k+1] = \mathbf{Ax}[k],
\label{eqn:plant}
\end{equation}
where $k \in \mathbb{N}$ is the discrete-time index, $\mathbf{x}[k] \in {\mathbb{R}}^n$ is the state vector, and  $\mathbf{A} \in {\mathbb{R}}^{ n \times n} $ is the system matrix. A network $\mathcal{G}=(\mathcal{V,E})$ of $N$ nodes monitor the state of this system. The $i$-th node receives a measurement of the state, given by
\begin{equation}
\mathbf{y}_{i}[k]=\mathbf{C}_i\mathbf{x}[k],
\label{eqn:Obsmodel}
\end{equation}
where $\mathbf{y}_{i}[k] \in {\mathbb{R}}^{r_i}$ and $\mathbf{C}_i \in {\mathbb{R}}^{r_i \times n}$. We define $\mathbf{C}\triangleq{\begin{bmatrix}\mathbf{C}^T_{1} \hspace{1.5mm} & \cdots & \hspace{1.5mm} \mathbf{C}^T_{N}\end{bmatrix}}^{T}$ and $\mathbf{y}[k]\triangleq {\begin{bmatrix}\mathbf{y}^T_{1}[k] \hspace{1.5mm} & \cdots & \hspace{1.5mm} \mathbf{y}^T_{N}[k]\end{bmatrix}}^T$ as the collective observation matrix, and collective measurement vector, respectively. In the standard distributed state estimation setup, each node $i$ is tasked with asymptotically recovering the entire state $\mathbf{x}[k]$. We make the basic (necessary) assumption  that the pair $(\mathbf{A},\mathbf{C})$ is detectable. However, for any given $i\in\mathcal{V}$, the pair $(\mathbf{A},\mathbf{C}_i)$ may not be detectable, thereby necessitating inter-node communications constrained by the topology of the network. 

\textbf{Diversity Model:} We capture node heterogeneity and, in particular, the fact that nodes have different vulnerabilities, by employing the notion of colors as suggested in \cite{faiq}. Specifically, let the set of colors be denoted $\Gamma=\{B_1,\ldots,B_{|\Gamma|}\}$, and let each node $i$ be assigned a unique color  $\Delta(i)$, where $\Delta(\cdot)$ is a mapping from $\mathcal{V} $ to $\Gamma$. Let the node set be partitioned accordingly as $\mathcal{V}=\{\mathcal{V}_{B_1},\ldots,\mathcal{V}_{B_{|\Gamma|}}\}.$ 

\textbf{Adversary Model:} We consider a subset $\mathcal{A} \subset \mathcal{V}$ of the nodes in the network to be adversarial; the remaining regular nodes will be denoted by the set $\mathcal{R}$.  The adversaries possess complete knowledge of the network topology, the system dynamics, and the algorithm employed by the non-adversarial nodes. They can act collaboratively, and can even transmit differing state estimates to different neighbors at the same instant of time, as per the Byzantine fault model \cite{Byz}. We require all adversarial nodes to be of the same type or color, i.e., the adversarial set is \textit{mono-chromatic}.\footnote{Our results can be easily generalized to account for a poly-chromatic adversarial set.} We do so to capture the impact of diverse node vulnerabilities: breach of a particular type of component (node) does not imply breach of the other types. We now recall the following definitions from \cite{broad} that quantify the number of adversaries in the network. 
\begin{definition} (\textbf{$f$-local set})
A set $\mathcal{C} \subset \mathcal{V}$ is \textit{$f$-local} if it contains at most $f$ nodes in the neighborhood of the other nodes, i.e., $|\mathcal{N}_i \cap \mathcal{C}| \leq f,$ $\forall i \in \mathcal{V}\setminus \mathcal{C}$.
\end{definition}
\begin{definition} (\textbf{$f$-local adversarial model}) A set $\mathcal{A}$ of adversarial nodes is \textit{$f$-locally bounded} if $\mathcal{A}$ is an $f$-local  set. 
\label{def:flocal}
\end{definition}

Within the class of mono-chromatic Byzantine adversarial models, we shall consider two sub-cases: one where the adversarial set $\mathcal{A}$ is $f$-locally bounded, and one where it is potentially not. We will refer to the former as the \textit{f-local mono-chromatic Byzantine adversary model}, and to the latter as simply the \textit{mono-chromatic Byzantine adversary model}. Each of these models has its own set of motivations and applies to different scenarios. For instance, the assumption of $f$-locality aims to account for scenarios where the adversary is resource-limited, and/or faces an increasing risk of getting detected with each component it compromises. When such considerations no longer apply, we relax the assumption of $f$-locality typically made in the literature on resilient distributed algorithms \cite{broad,vaidyacons,rescons,dibaji,abbas,faiq,usevitch1,mitra_auto,mitraAR,Sundaramopt,su,Byz,mitra2018impact}, and  allow an adversary to compromise an \textit{arbitrary} number of nodes of a particular type. Our philosophy here is as follows: once an adversary has figured out a way to breach the security of a particular type of component (node), it is in its interest to compromise more (if not all) nodes of that type, if this does not incur any additional resource or risk on its part (for example, malware and viruses).  

 Finally, let us note that the actual number and identities of the adversarial nodes are not known to the regular nodes. We do, however, assume that each regular node is aware of (i) the true color of each of its neighbors, including those that are adversarial; and (ii) the upper-bound $f$ on the number of adversaries in its neighborhood, whenever $\mathcal{A}$ is $f$-local.

\textbf{Trust Model:} We assume that a subset $\mathcal{T}\subseteq\mathcal{V}$ of nodes cannot be compromised by adversaries, i.e., $\mathcal{T}\cap\mathcal{A}=\emptyset$. Furthermore, we assume that each node is aware of the identities of its trusted neighbors. Note that when $|\Gamma|=1$, i.e., when all nodes are of the same type, we recover the setting in \cite{mitra2018impact}, where only the impact of trusted nodes was considered.  

With all the relevant models set up, we are now in position to state the problem of interest. To this end, let  $\hat{\mathbf{x}}_i[k]$ represent the estimate of $\mathbf{x}[k]$ (the state of system \eqref{eqn:plant})  maintained by node $i$. Our objective in this paper will be to study how  diversity and trust can be exploited to solve the following problem.

\begin{problem} 
\label{prob:resdist}
\textbf{(Resilient Distributed State Estimation)} Given an LTI system (\ref{eqn:plant}), a linear measurement model (\ref{eqn:Obsmodel}), and a time-invariant directed communication graph $\mathcal{G}$, design a set of state estimate update and information exchange rules such that $\lim_{k\to\infty} \Vert\hat{\mathbf{x}}_i[k]-\mathbf{x}[k]\Vert=0$, $\forall i \in \mathcal{R}$, \textit{regardless} of the actions of any $f$-local mono-chromatic set of Byzantine adversaries.\footnote{Later, in Section \ref{sec:mono_chrom_adv}, we investigate a variant of Problem \ref{prob:resdist} where the $f$-locality assumption on the adversarial model is relaxed.}
\end{problem}

\section{Resilient Distributed State Estimation under mono-chromatic Byzantine adversaries}
\label{sec:floc_graph_algo}
\subsection{Characterizing Sufficient Graph-theoretic Conditions}
\label{subsec:floc_graph}
In this section, we identify certain graph-theoretic conditions that play a key role in our proposed solution to Problem \ref{prob:resdist}. In particular, these topological conditions are sufficient to solve Problem \ref{prob:resdist} based on an approach that we develop later in Section \ref{subsec: floc_algo}. To proceed, we introduce the following notion of $(r,\Delta(\cdot),\mathcal{T})$-reachability. 

\begin{definition}($(r,\Delta(\cdot),\mathcal{T})$-\textbf{reachable set}) Consider a graph $\mathcal{G}=(\mathcal{V,E})$ with a trusted node set $\mathcal{T}$, where each node $i\in\mathcal{V}$ is assigned a color $\Delta(i)$. Then, given $r\in\mathbb{N}_{+}\cup\{\infty\}$, and a non-empty set $\mathcal{C}\subseteq\mathcal{V}$, $\mathcal{C}$ is said to be an $(r,\Delta(\cdot),\mathcal{T})$-reachable set if $\exists i\in\mathcal{C}$ satisfying at least one of the following conditions:
\begin{enumerate}
    \item[(i)] \textbf{Redundancy}: Node $i$ has at least $r$ neighbors outside $\mathcal{C}$, i.e., $|\mathcal{N}_i\setminus\mathcal{C}| \geq r.$ 
    \item[(ii)] \textbf{Diversity}: Node $i$ has at least 3 distinct colored neighbors outside $\mathcal{C}$, i.e., there exist nodes $u,v,w\in\mathcal{N}_i\setminus\mathcal{C}$, such that $\Delta(u)\neq\Delta(v)\neq\Delta(w)\neq\Delta(u)$. 
    \item[(iii)] \textbf{Trust}: Node $i$ has at least one trusted neighbor outside $\mathcal{C}$, i.e., $|\{\mathcal{N}_i\setminus\mathcal{C}\}\cap\mathcal{T}| \geq 1.$ 
\end{enumerate}
\label{defn:r_delta_tau_reachable}
\end{definition}

When $r=\infty$ (as in Section \ref{sec:mono_chrom_adv}), the above definition will correspond to that of a $(\Delta(\cdot),\mathcal{T})$ reachable set. The conditions in Defn. \ref{defn:r_delta_tau_reachable} are illustrated in Figure \ref{fig:val_con}. 

\begin{figure}[h]
    \centering
    \includegraphics[scale=1]{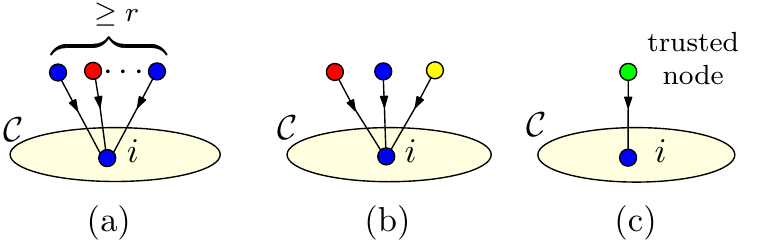}
    \caption{{Illustration of sets that satisfy $(r,\Delta(\cdot),\mathcal{T})$-reachability as per Defn. \ref{defn:r_delta_tau_reachable}, via (a) redundancy, (b) diversity, or (c) trust.}}
    \label{fig:val_con}
\end{figure}

Next, we introduce the key topological property required to solve Problem \ref{prob:resdist} based on our proposed approach. 
\begin{definition}(\textbf{strongly} $(r,\Delta(\cdot),\mathcal{T})$-\textbf{robust graph \textit{w.r.t.} $\mathcal{S}$}) Consider a graph $\mathcal{G}=(\mathcal{V,E})$ with a trusted node set $\mathcal{T}$, where each node $i\in\mathcal{V}$ is assigned a color $\Delta(i)$. Then, given $r\in\mathbb{N}_{+}\cup\{\infty\}$, and a set $\mathcal{S}\subseteq\mathcal{V}$,  $\mathcal{G}$ is strongly $(r,\Delta(\cdot),\mathcal{T})$-robust w.r.t. $\mathcal{S}$ if for all non-empty subsets  $\mathcal{C} \subseteq \mathcal{V}\setminus\mathcal{S}$, $\mathcal{C}$ is $(r,\Delta(\cdot),\mathcal{T})$-reachable.
\label{defn:strong_rdeltatau}
\end{definition}
When all nodes are of the same color, i.e., when $\Delta(i)=\Delta(j), \forall i,j\in\mathcal{V}$, and when the trusted set $\mathcal{T}$ is empty, we recover the conventional notions of $r$-reachability \cite{rescons}, and strong $r$-robustness w.r.t. a set $\mathcal{S}$ \cite{mitra_auto}, from Defn.'s  \ref{defn:r_delta_tau_reachable} and  \ref{defn:strong_rdeltatau}, respectively. 
{We note that the notion of strong $(r,\Delta(\cdot),\mathcal{T})$-robustness realizes the idea that there are multiple ways to achieve a desired level of robustness in the underlying network: by creating extra links between nodes (redundancy), or by diversifying nodes (diversity), or by hardening a subset of the nodes (trust), or by a combination of these approaches. For instance, consider the graph in Figure \ref{fig:comparison}(a), in which all nodes have the same color (no diversity), and there is no trusted node. The graph is strongly $(3,\Delta(\cdot),\mathcal{T})$-robust w.r.t. $\mathcal{S} = \{1,2,3,4,5,6\}$, where $\Delta(i) = \Delta(j)$, $\forall i\ne j$, and $\mathcal{T} = \emptyset$. We can make such a graph strongly $(6,\Delta(\cdot),\mathcal{T})$-robust w.r.t. $\mathcal{S}$ simply by adding extra links between nodes as shown in Figure \ref{fig:comparison}(b). At the same time, if we have three colors, then we can assign them to nodes such that the graph becomes strongly ($6$,$\Delta(\cdot),\mathcal{T}$)-robust w.r.t. $\mathcal{S}$, without adding extra edges or trusted nodes, as shown in Figure \ref{fig:comparison}(c). Similarly, if node 4 is a trusted node, while all the remaining nodes are of the same color, the graph again becomes strongly ($6$,$\Delta(\cdot),\mathcal{T}$)-robust w.r.t. $\mathcal{S}$, with no extra edges, as illustrated in Fig. \ref{fig:comparison}(d).}

\begin{figure}[h!]
\centering
\begin{subfigure}[b]{0.25\textwidth}
\centering
\includegraphics[scale=0.34]{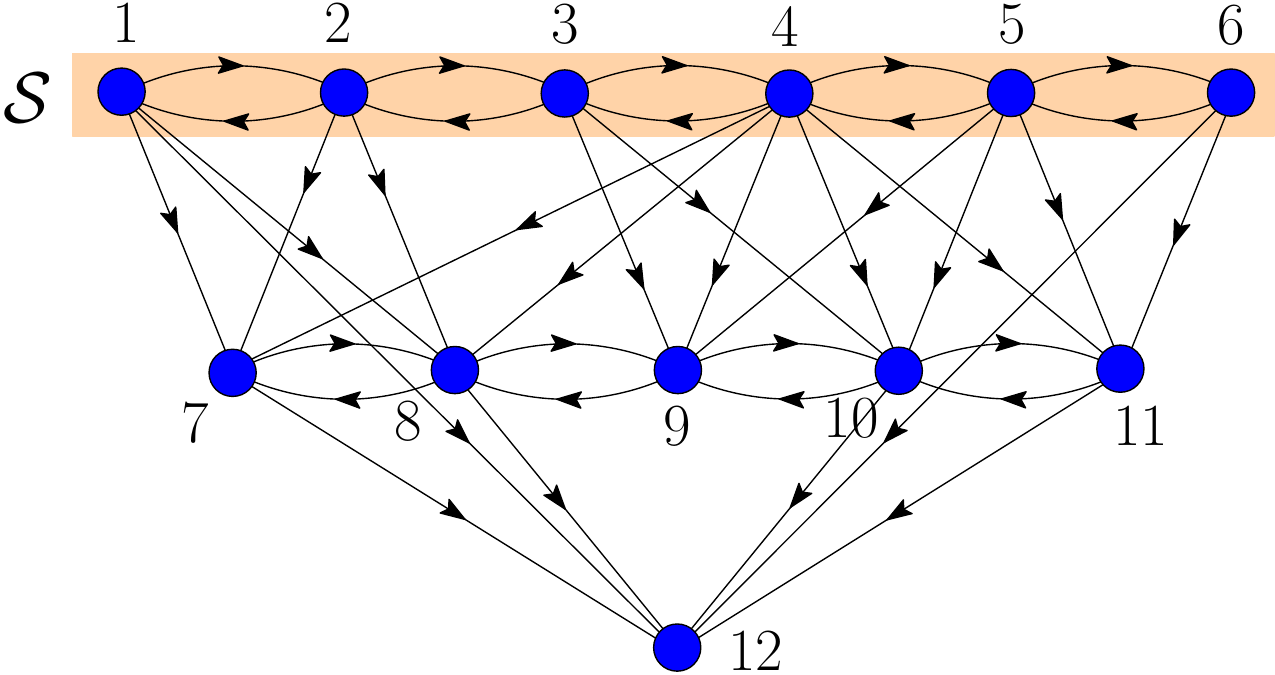}
\caption{}
\end{subfigure}
\begin{subfigure}[b]{0.22\textwidth}
\centering
\includegraphics[scale=0.34]{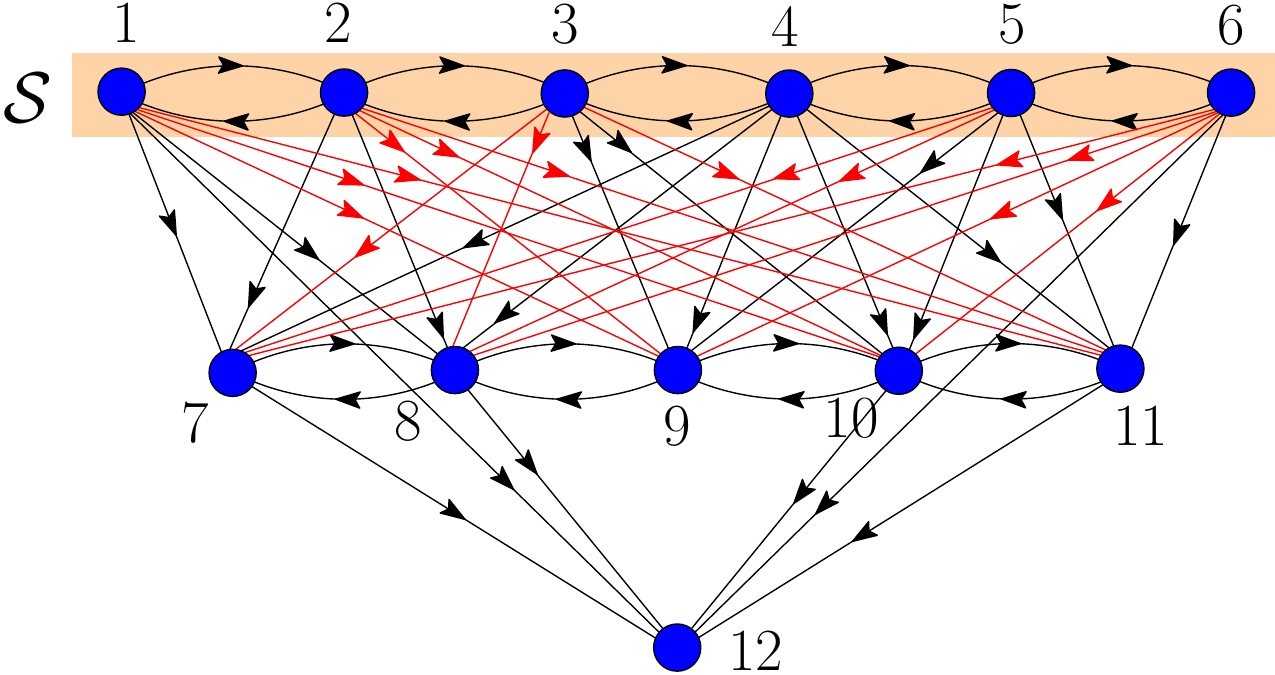}
\caption{}
\end{subfigure}
\begin{subfigure}[b]{0.25\textwidth}
\centering
\includegraphics[scale=0.34]{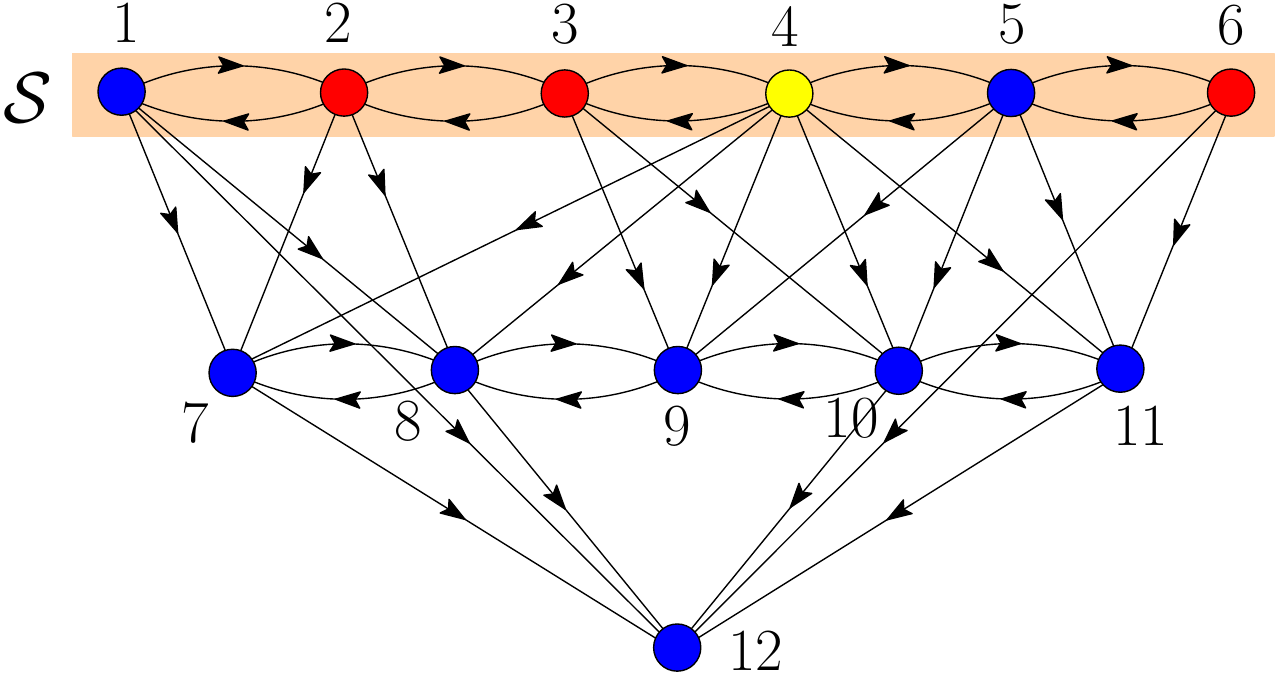}
\caption{}
\end{subfigure}
\begin{subfigure}[b]{0.22\textwidth}
\centering
\includegraphics[scale=0.34]{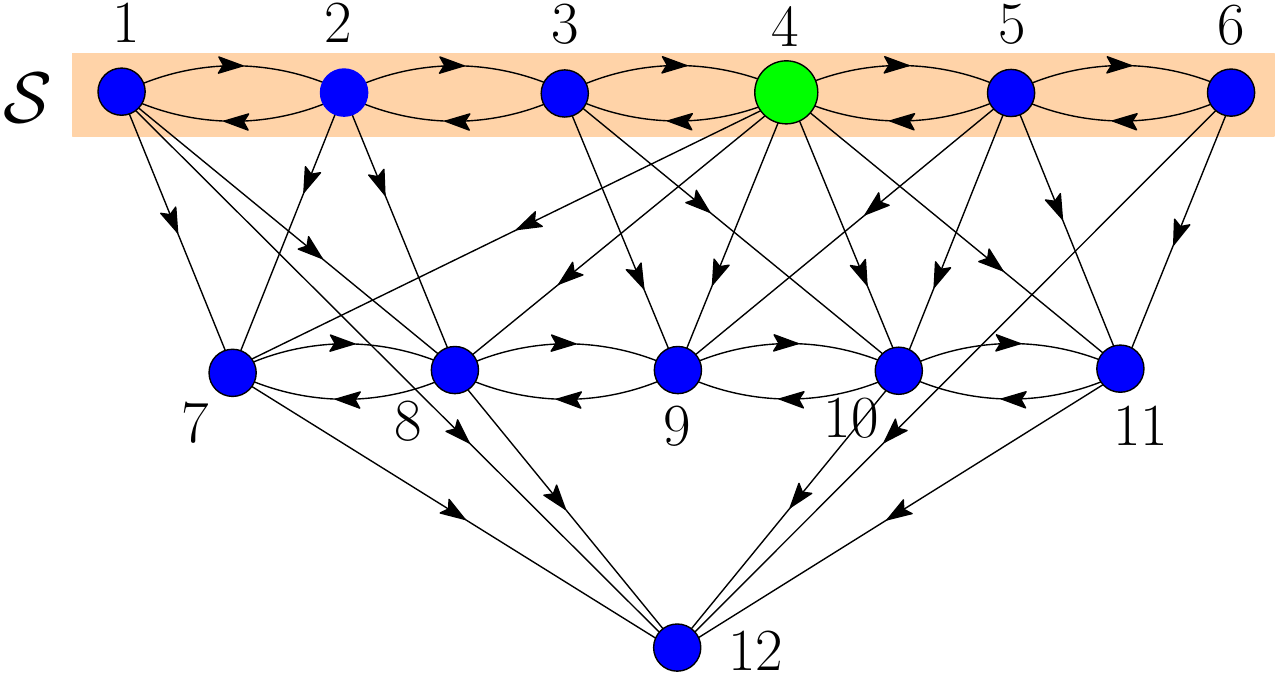}
\caption{}
\end{subfigure}
\caption{{Illustration of different approaches, including redundancy (b), diversity (c), and trust (d), to improve network robustness. The set of source nodes in all figures is $\mathcal{S} =~ \{1,2,\ldots,6\}$. Node 4 is a trusted node in Fig. \ref{fig:comparison}(d). The graphs in Figs. \ref{fig:comparison}(b), \ref{fig:comparison}(c), and  \ref{fig:comparison}(d) are all strongly $(6,\Delta(\cdot),\mathcal{T})$-robust w.r.t. $\mathcal{S}$.}}
\label{fig:comparison}
\end{figure}

Next, we recall the notion of source nodes \cite{mitra_auto}. 

\begin{definition}
(\textbf{Source nodes})
For each $\lambda_j \in \Lambda_{U}(\mathbf{A})$, let the set $\mathcal{S}_j$ be defined as follows:
\begin{equation}
\mathcal{S}_j\triangleq\{i\in\mathcal{V}|\textrm{rank}\begin{bmatrix}\mathbf{A}-\lambda_j\mathbf{I}_n\\ \mathbf{C}_i\end{bmatrix}=n \}.
\end{equation}
Then, $\mathcal{S}_j$ will be called the set of source nodes for $\lambda_j$.\footnote{In case $i\in\mathcal{S}_j$, we will say that ``node $i$ can detect $\lambda_j$". Each stable eigenvalue is considered detectable w.r.t. the measurements of every node.}
\end{definition}

\begin{figure}[t]
 \centering
\begin{subfigure}[b]{0.22\textwidth}
\centering
\includegraphics[scale=0.825]{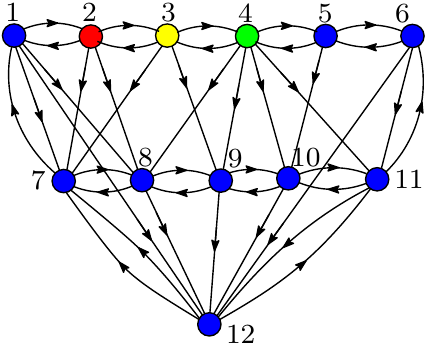}
\caption{$\mathcal{G}$}
\end{subfigure}
\begin{subfigure}[b]{0.25\textwidth}
\centering
\includegraphics[scale=0.825]{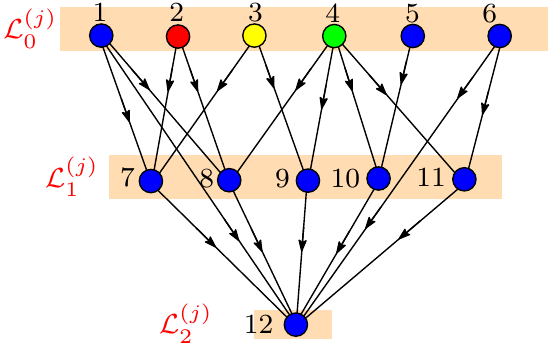}
\caption{$\mathcal{G}_j$}
\end{subfigure}
\caption{{MEDAG illustration with $f=1$. $\mathcal{S}_j = \{1,2,\ldots,6\}$ is the set of source nodes, and node $4$ is trusted. Among the non-source nodes in $\mathcal{G}_j$, node $7$ in $\mathcal{L}_1^{(j)}$ has three distinct colored neighbors in $\mathcal{L}_0^{(j)}$ (diversity condition), whereas each of the nodes in $\{8,9,10,11\}$ has a trusted neighbor in $\mathcal{L}_0^{(j)}$ (trust condition). At the same time, node 12 in $\mathcal{L}_2^{(j)}$ has more than $(2f+1)$ neighbors in $\mathcal{L}_0^{(j)}\cup \mathcal{L}_1^{(j)}$ (redundancy condition).}}
\label{fig:MEDAG2}
\end{figure}
 
Let $\Omega_{U}(\mathbf{A}) \subseteq \Lambda_{U}(\mathbf{A})$ contain the set of eigenvalues of $\mathbf{A}$ for which $\mathcal{V}\setminus{\mathcal{S}_j}$ is non-empty. Essentially, for each unstable mode $\lambda_j\in\Omega_{U}(\mathbf{A})$, the source nodes $\mathcal{S}_j$ can leverage their own local measurements to estimate the portion of the state corresponding to $\lambda_j$. However, to enable each non-source node $i\in \mathcal{V}\setminus\mathcal{S}_j$ to estimate that portion, a secure medium of information flow from $\mathcal{S}_j$ to $\mathcal{V}\setminus\mathcal{S}_j$ is necessary. To this end, the concept of a Mode Estimation Directed Acyclic Graph (MEDAG) was introduced in \cite{mitra_auto}. We now suitably modify the definition of a MEDAG to account for diversity and trust. 

\begin{definition}
\label{defn:MEDAG_r,delta,T)}
(\textbf{ $(2f+1,\Delta(\cdot),\mathcal{T})$ Mode Estimation Directed Acyclic Graph (MEDAG)}) Consider a mode $\lambda_j \in \Omega_{U}(\mathbf{A})$. Suppose there exists a spanning sub-graph $\mathcal{G}_j = (\mathcal{V},\mathcal{E}_j)$ of $\mathcal{G}$ with the following properties for all $f$-local, mono-chromatic sets $\mathcal{A}$ with $\mathcal{A}\cap\mathcal{T}=\emptyset$, and $\mathcal{R}=\mathcal{V}\setminus\mathcal{A}$.
\begin{itemize}
\item[(i)] If $i \in \{\mathcal{V}\setminus\mathcal{S}_j\} \cap \mathcal{R}$, then either $|\mathcal{N}^{(j)}_i| \geq 2f+1$; or $|\mathcal{N}^{(j)}_i\cap\mathcal{T}| \geq 1$; or $\exists u,v,w\in\mathcal{N}^{(j)}_i$ such that $\Delta(u)\neq\Delta(v)\neq\Delta(w)\neq\Delta(u)$. Here, $\mathcal{N}^{(j)}_i=\{l\in\mathcal{V}|(l,i) \in \mathcal{E}_j\}$ represents the neighborhood of node $i$ in $\mathcal{G}_j$.
\item[(ii)] There exists a partition of $\mathcal{R}$ into sets $\{\mathcal{L}^{(j)}_{0}, \ldots , \mathcal{L}^{(j)}_{T_j}\}$, where $T_j \in\{0,\ldots,N-1\}$,  $\mathcal{L}^{(j)}_{0} = \mathcal{S}_j \cap \mathcal{R}\neq\emptyset$, and if $i \in \mathcal{L}^{(j)}_q$ (where $1 \leq q \leq T_j)$, then $\mathcal{N}^{(j)}_i \cap \mathcal{R} \subseteq \bigcup^{q-1}_{r=0} \mathcal{L}^{(j)}_r$. Furthermore, $\mathcal{N}^{(j)}_i=\emptyset, \forall i \in \mathcal{L}^{(j)}_0$.
\end{itemize}
Then, we call $\mathcal{G}_j$ a  $(2f+1,\Delta(\cdot),\mathcal{T})$ MEDAG for $\lambda_j$. 
\end{definition}

In the above definition, condition (i) requires each non-source node in $\mathcal{R}$ to either have $(2f+1)$ neighbors, or a trusted neighbor, or three distinct colored neighbors in $\mathcal{G}_j$. Condition (ii) in turn states that in $\mathcal{G}_j$, the  set $\mathcal{R}$ should admit a partition into levels $\{\mathcal{L}^{(j)}_{0}, \ldots , \mathcal{L}^{(j)}_{T_j}\}$, such that a node in a particular level $q$ has neighbors in $\mathcal{R}$ from levels strictly lower than $q$, leading to an acyclic structure. An example of a $(2f+1,\Delta(\cdot),\mathcal{T})$ MEDAG is shown in Fig.  \ref{fig:MEDAG2} for $f = 1$.

\textbf{Construction of a $(2f+1,\Delta(\cdot),\mathcal{T})$ MEDAG}: We briefly discuss an algorithm that can be used to construct a $(2f+1,\Delta(\cdot),\mathcal{T})$ MEDAG (conditions for the existence of such a MEDAG will be provided below). Suppose we are given a graph $\mathcal{G}=(\mathcal{V,E})$ with a trusted node set $\mathcal{T}$, where each node $i\in\mathcal{V}$ is assigned a color $\Delta(i)$. For each $\lambda_j\in\Omega_{U}(\mathbf{A})$, our objective is to construct a sub-graph $\mathcal{G}_j$ satisfying the conditions in Defn.  \ref{defn:MEDAG_r,delta,T)} and, in the process, to identify the sets $\mathcal{N}^{(j)}_i, \, \forall i \in \mathcal{V}$. With the sets $\mathcal{N}^{(j)}_i$ in hand, one can implement the resilient distributed state estimation algorithm to be described later in Section \ref{subsec: floc_algo}. The MEDAG construction algorithm requires each node $i$ to maintain a counter $c_i(j)$ and a list of indices $\mathcal{N}^{(j)}_i$ for each $\lambda_j \in \Omega_{U}(\mathbf{A})$. These parameters are initialized with $c_i(j)=0$ and $\mathcal{N}^{(j)}_i = \emptyset$, for each $i\in\mathcal{V}$. Subsequently, the algorithm proceeds in rounds where in round zero, each node in $\mathcal{S}_j$ broadcasts the message $``1"$ to its out-neighbors, sets $c_i(j)=1$, maintains $\mathcal{N}^{(j)}_i = \emptyset$ for all future rounds, and goes to sleep. A node $i \in \mathcal{V}\setminus\mathcal{S}_j$ waits until it either receives $``1"$ from at least $(2f+1)$ distinct neighbors, or from at least three distinct colored neighbors, or from at least one
 trusted neighbor. When any one of these conditions is eventually met, it sets $c_i(j)=1$, appends the labels of each of the neighbors from which it received $``1"$ to $\mathcal{N}^{(j)}_i$, broadcasts the message $``1"$ to its out-neighbors, and goes to sleep. The MEDAG construction algorithm ``\textit{terminates for $\lambda_j$}" if there exists $T_j \in \mathbb{N}_{+}$ such that $c_i(j) =1$ $\forall i \in \mathcal{V}$, for all rounds following round $T_j$. The \textbf{objective} of the algorithm is to return a set of sets $\{\mathcal{N}^{(j)}_i\}$, where $\lambda_j \in \Omega_{U}(\mathbf{A})$, $i \in \mathcal{V}$.\footnote{Here, we do not consider adversarial behaviour during the MEDAG construction phase; however, such a possibility can be readily accounted for following arguments similar to those in \cite{mitra_auto}.} 

In the following result, we establish that the notion of strong $(r,\Delta(\cdot),\mathcal{T})$-robustness tightly characterizes the existence of a MEDAG as described in Defn.  \ref{defn:MEDAG_r,delta,T)}. In Section \ref{subsec: floc_algo}, we will demonstrate how the existence of such sub-graphs features in the synthesis of our resilient filtering algorithm. 

\begin{theorem} For each $\lambda_j \in \Omega_{U}(\mathbf{A})$, $\mathcal{G}$ contains a $(2f+1,\Delta(\cdot),\mathcal{T})$ MEDAG for $\lambda_j$ if and only if $\mathcal{G}$ is strongly $(2f+1,\Delta(\cdot),\mathcal{T})$-robust w.r.t. $\mathcal{S}_j$.
\label{thm:MEDAG_existence}
\end{theorem}
\begin{proof}
``$\Longleftarrow$" Consider any $\lambda_j \in \Omega_{U}(\mathbf{A})$, and suppose that $\mathcal{G}$ is strongly $(2f+1,\Delta(\cdot),\mathcal{T})$-robust w.r.t. $\mathcal{S}_j$. We argue that the MEDAG construction algorithm described in this section terminates and, upon termination, returns a set of neighbor relations  $\{\mathcal{N}^{(j)}_i\}$ that induce a sub-graph $\mathcal{G}_j$ satisfying each of the two properties outlined in Defn. \ref{defn:MEDAG_r,delta,T)}. To this end, let the set of nodes that get ``activated" during the $q$-th round of the MEDAG construction algorithm be denoted $\mathcal{C}^{(j)}_q$, where we say that a node $i$ is activated as soon as it sets $c_i(j)$ to 1. Then, based on Defn.'s \ref{defn:r_delta_tau_reachable} and \ref{defn:strong_rdeltatau}, it is easy to see that $\mathcal{C}^{(j)}_0=\mathcal{S}_j$ is non-empty. By way of contradiction, suppose the MEDAG construction algorithm does not terminate. This implies the existence of a non-empty set $\mathcal{P}\subseteq\mathcal{V}\setminus{\mathcal{S}}_j$ of nodes that never get activated. However, since $\mathcal{P}$ is $(2f+1,\Delta(\cdot),\mathcal{T})$-reachable, each node $i\in\mathcal{P}$ must have received ``1" from either $(2f+1)$ nodes outside $\mathcal{P}$, or from 3 distinct colored nodes outside $\mathcal{P}$, or from a trusted node outside $\mathcal{P}$, causing it to get activated. This leads to the desired contradiction, and we conclude that all nodes must get activated eventually. It is easy to see that such an activation process can take at most $N-1$ rounds, since each new round activates at least one new node until the time all nodes get activated. Thus, there must exist some $T_j \leq N-1$ such that $\bigcup^{T_j}_{q=0}\mathcal{C}^{(j)}_q=\mathcal{V}$. Now consider any $f$-local, mono-chromatic set $\mathcal{A}$ satisfying $\mathcal{A}\cap\mathcal{T}=\emptyset$, and let $\mathcal{R}=\mathcal{V}\setminus\mathcal{A}$. For each $q\in\{0,\ldots,T_j\}$, define $\mathcal{L}^{(j)}_q\triangleq\mathcal{C}^{(j)}_q\cap\mathcal{R}$. Since $\{\mathcal{C}^{(j)}_q\}^{T_j}_{q=0}$ partitions $\mathcal{V}$, $\{\mathcal{L}^{(j)}_q\}^{T_j}_{q=0}$ partitions $\mathcal{R}$. Since upon activation, a node goes to sleep and does not listen to nodes that get subsequently activated, we have that for any $q\in\{1,\ldots,T_j\}$, if $i\in\mathcal{C}^{(j)}_q$, then $\mathcal{N}^{(j)}_i\subseteq\bigcup_{r=0}^{q-1}\mathcal{C}^{(j)}_r$. Thus, if $i\in\mathcal{L}^{(j)}_q$, then $\mathcal{N}^{(j)}_i\cap\mathcal{R}\subseteq\bigcup_{r=0}^{q-1}\mathcal{L}^{(j)}_r$. To verify property (ii) in Defn. \ref{defn:MEDAG_r,delta,T)}, it remains to argue that $\mathcal{L}^{(j)}_0=\mathcal{S}_j\cap\mathcal{R}\neq\emptyset$. Assume to the contrary that $\mathcal{S}_j\cap\mathcal{R}=\emptyset$, i.e., $\mathcal{S}_j\subseteq\mathcal{A}$. Thus, $\mathcal{R}\subseteq\mathcal{V}\setminus\mathcal{S}_j$, and at the same time $\mathcal{R}$ is non-empty since $\mathcal{A}$ is $f$-local (see Defn. \ref{def:flocal}). Since $\mathcal{G}$ is strongly $(2f+1,\Delta(\cdot),\mathcal{T})$-robust w.r.t. $\mathcal{S}_j$, it must then be that $\mathcal{R}$ is $(2f+1,\Delta(\cdot),\mathcal{T})$-reachable - a condition that is impossible to satisfy given the fact that $\mathcal{A}$ is $f$-local, mono-chromatic and $\mathcal{A}\cap\mathcal{T}=\emptyset.$ This leads to the desired contradiction, establishing property (ii) in Defn. \ref{defn:MEDAG_r,delta,T)}. Now consider any $i\in\{\mathcal{V}\setminus\mathcal{S}_j\} \cap \mathcal{R}$, and note that it must belong to some $\mathcal{L}^{(j)}_q$, where $q\in\{1,\ldots,T_j\}$. Thus, it must get activated at some point, and property (i) in Defn. \ref{defn:MEDAG_r,delta,T)} follows by simply noting the conditions for activation of a node in the MEDAG construction algorithm. 

``$\Longrightarrow$" We prove necessity via contradiction. Given some $\lambda_j\in\Omega_{U}(\mathbf{A})$, let there exist a  sub-graph $\mathcal{G}_j$ satisfying the two properties in Defn. \ref{defn:MEDAG_r,delta,T)}. Suppose $\mathcal{G}$ is not strongly $(2f+1,\Delta(\cdot),\mathcal{T})$-robust w.r.t. $\mathcal{S}_j$. Thus, there exists a non-empty set $\mathcal{C} \subseteq \mathcal{V}\setminus\mathcal{S}_j$ that is not $(2f+1,\Delta(\cdot),\mathcal{T})$-reachable. Consider the trivial $f$-local set $\mathcal{A}=\emptyset$ that satisfies $\mathcal{A} \cap \mathcal{T}=\emptyset$.\footnote{Here, we adhere to the convention that an empty set is mono-chromatic.} The sub-graph $\mathcal{G}_j$ must contain a partition of $\mathcal{R}=\mathcal{V}\setminus\mathcal{A}=\mathcal{V}$ into sets $\{\mathcal{L}^{(j)}_q\}^{T_j}_{q=0}$ that satisfy property (ii) in Defn. \ref{defn:MEDAG_r,delta,T)}. Accordingly, let $\mathcal{C}$ get partitioned as  $\mathcal{C}=\bigcup_{q=1}^{T_j} \mathcal{F}_q$, where $\mathcal{F}_q=\mathcal{C}\cap\mathcal{L}^{(j)}_{q}$ (note that $\mathcal{C}\cap\mathcal{L}^{(j)}_0=\emptyset$). Let $p$ be the smallest integer such that $\mathcal{F}_p$ is non-empty. Then, from property (ii) in Defn. \ref{defn:MEDAG_r,delta,T)}, it follows that for any $i \in \mathcal{F}_{p}$, $\mathcal{N}^{(j)}_i$ contains elements from only $\mathcal{V}\setminus\mathcal{C}$. However, as $\mathcal{C}$ is not $(2f+1,\Delta(\cdot),\mathcal{T})$-reachable, $\mathcal{N}^{(j)}_i$ violates each of the three conditions in property (i) of Defn. \ref{defn:MEDAG_r,delta,T)}, leading to the desired contradiction. 
\end{proof}

\subsection{Algorithm and Analysis for $f$-local Mono-chromatic Byzantine Adversaries}
\label{subsec: floc_algo}

In this section, we develop an algorithm that leverages node-diversity and  trusted nodes to solve Problem \ref{prob:resdist}. For clarity of exposition, we make the following assumption on the system matrix $\mathbf{A}$.
\begin{assumption}
$\mathbf{A}$ has real, distinct eigenvalues.
\label{ass:system}
\end{assumption}
Although the above assumption might seem restrictive, the results that we derive subsequently can be generalized to account for system matrices with arbitrary spectrum using a more detailed technical analysis as in \cite{mitra_auto}. Since any $\mathbf{A}$ satisfying Assumption \ref{ass:system} can be diagonalized via an appropriate similarity transformation, we assume without loss of generality that $\mathbf{A}$ is already in diagonal form. Specifically, suppose $\mathbf{A}=diag(\lambda_1,\cdots,\lambda_n)$, where $sp(\mathbf{A})=\{\lambda_1,\ldots,\lambda_n\}$. Let the component of the state vector $\mathbf{x}[k]$ corresponding to eigenvalue $\lambda_j$ be denoted by $x^{(j)}[k]$. Building on the general idea developed in \cite{mitra_auto}, for each $\lambda_j \in \Omega_{U}(\mathbf{A})$, the source nodes $\mathcal{S}_j$ and the non-source nodes $\mathcal{V}\setminus\mathcal{S}_j$ employ separate update rules for estimating $x^{(j)}[k]$. In particular, the source nodes maintain local\footnote{Here, by `local',  we imply that such observers can be constructed and run without any information from neighbors.} Luenberger observers for estimating $x^{(j)}[k]$, while the non-source nodes rely on a resilient consensus based protocol to achieve this task. For any node $i$, let the set of eigenvalues it can detect be denoted by $\mathcal{O}_i$, and let $\bar{\mathcal{O}}_i=sp(\mathbf{A})\setminus\mathcal{O}_i$. Then, the following result from \cite{mitra_auto} states that node $i$ can estimate the components of $\mathbf{x}[k]$ corresponding to the eigenvalues in $\mathcal{O}_i$, (i.e., the locally detectable portion of $\mathbf{x}[k]$) \textit{without} interacting with its neighbors. 

\begin{lemma}
Suppose Assumption \ref{ass:system} holds. Then, for each $i\in\mathcal{R}$, a local Luenberger observer can be constructed that ensures $\lim_{k\to\infty}|\hat{x}^{(j)}_i[k]-x^{(j)}[k]|=0, \forall \lambda_j \in \mathcal{O}_{i}$,  where $\hat{x}^{(j)}_i[k]$ denotes the estimate of $x^{(j)}[k]$ maintained by node $i$.
\label{lemma:luen}
\end{lemma}

In what follows, we develop a filtering algorithm that allows each regular node to estimate the locally undetectable portion of the dynamics, despite the potential presence of adversarial nodes in its neighborhood. 
The proposed filtering algorithm, adapted to account for node-diversity and the presence of trusted nodes, involves the following steps. 

For each $\lambda_j\in\bar{\mathcal{O}}_i$, $i\in \mathcal{R}$ updates $\hat{x}^{(j)}_i[k]$ as follows.
\begin{itemize}
\item[1)] At each time-step $k$, node $i$ collects estimates of $x^{(j)}[k]$ received from \textit{only} those neighbors that belong to $\mathcal{N}^{(j)}_i \subseteq \mathcal{N}_i$. Recall that $\mathcal{N}^{(j)}_i$ represents neighbors of node $i$ in the MEDAG $\mathcal{G}_j$ (see Definition \ref{defn:MEDAG_r,delta,T)}).

\item [2)] If $\mathcal{N}^{(j)}_i \cap\mathcal{T} \neq \emptyset$, then $\hat{x}^{(j)}_i[k]$ is updated as follows:
\begin{equation}
\hat{x}^{(j)}_i[k+1]=\lambda_j\left(\sum_{l \in \mathcal{N}^{(j)}_i\cap\mathcal{T}}\bar{w}^{(j)}_{il}\hat{x}^{(j)}_{l}[k]\right),
\label{eqn:LFRE_1}
\end{equation}
where the weights $\bar{w}^{(j)}_{il}$ are non-negative and chosen to satisfy $\sum_{l \in \mathcal{N}^{(j)}_i\cap\mathcal{T}}\bar{w}^{(j)}_{il}=1$.

\item [3)] If $\mathcal{N}^{(j)}_i \cap\mathcal{T} = \emptyset$, but there exist three distinct colored nodes in $\mathcal{N}^{(j)}_i$, then node $i$ sorts the estimates of $x^{(j)}[k]$ received from $\mathcal{N}^{(j)}_i$ in descending order. Upon such sorting, let the indices of the nodes in $\mathcal{N}^{(j)}_i$ be $\{n_1,\ldots,n_{|\mathcal{N}^{(j)}_i|}\}$, i.e., $\hat{x}^{(j)}_{n_1}[k] \geq \hat{x}^{(j)}_{n_2}[k] \ldots \geq \hat{x}^{(j)}_{n_{|\mathcal{N}^{(j)}_i|}}[k].$\footnote{Here, we have suppressed the dependence of the indices $n_p$ on $i,j$ and $k$ for clarity of exposition.} Define $m\triangleq \min\{p:\Delta(n_p)\neq\Delta(n_1)\},$ and $M\triangleq \max\{p:\Delta(n_p)\neq\Delta(n_{|\mathcal{N}^{(j)}_i|})\}$. It can be easily verified that, when $\mathcal{N}^{(j)}_i$ contains at least 3 distinct colored nodes, we have $M\geq m$. Accordingly, let $\mathcal{R}^{(j)}_i[k]=\cup_{p=m}^{M} n_{p}$. Then, $\hat{x}^{(j)}_i[k]$ is updated as follows:
\begin{equation}
\hat{x}^{(j)}_i[k+1]=\lambda_j\left(\sum_{l \in \mathcal{R}^{(j)}_i[k]}\tilde{w}^{(j)}_{il}[k]\hat{x}^{(j)}_{l}[k]\right),
\label{eqn:LFRE_2}
\end{equation}
where the weights $\tilde{w}^{(j)}_{il}[k]$ are non-negative and chosen to satisfy $\sum_{l\in\mathcal{R}^{(j)}_i[k]}\tilde{w}^{(j)}_{il}[k]=1$.\footnote{In  words, from each end, node $i$ keeps rejecting estimates until it encounters a node with color different from that of the node with the most extreme estimate on that end. See Fig. \ref{fig:LFRE}(b) for an illustration of this step.}

\item [4)] If $\mathcal{N}^{(j)}_i \cap\mathcal{T} = \emptyset$, and node $i$ does not contain three distinct colored neighbors in $\mathcal{N}^{(j)}_i$, then it first sorts the estimates of $x^{(j)}[k]$ received from $\mathcal{N}^{(j)}_i$ in descending order, just as in Step 3. It then removes the highest and lowest $f$ estimates (i.e., removes $2f$ estimates in all), and updates $\hat{x}^{(j)}_i[k]$ as follows:
\begin{equation}
\hat{x}^{(j)}_i[k+1]=\lambda_j\left(\sum_{l \in \mathcal{M}^{(j)}_i[k]}w^{(j)}_{il}[k]\hat{x}^{(j)}_{l}[k]\right),
\label{eqn:LFRE_3}
\end{equation}
where $\mathcal{M}^{(j)}_i[k] \subset \mathcal{N}^{(j)}_i (\subseteq \mathcal{N}_i)$ is the set of nodes from which node $i$ chooses to accept estimates of $x^{(j)}[k]$ at time-step $k$, after removing the $f$ highest and $f$ lowest estimates from $\mathcal{N}^{(j)}_i$. The weights $w^{(j)}_{il}[k]$ are non-negative and chosen to satisfy $\sum_{l \in \mathcal{M}^{(j)}_i[k]}{w^{(j)}_{il}}[k]=1$.
\end{itemize}
We refer to the above algorithm as the Local-Filtering based Resilient Estimation (LFRE) algorithm for $f$-local mono-chromatic Byzantine adversaries; the steps of this algorithm are illustrated in Fig. \ref{fig:LFRE}. The following key result of our paper demonstrates how redundancy, diversity, and trust can be leveraged to perform resilient distributed state estimation.

\begin{figure}
    \centering
    \includegraphics[scale=1]{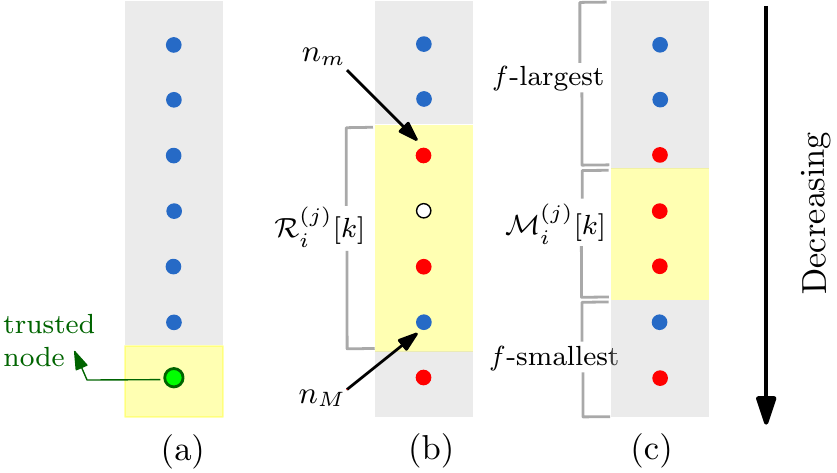}
    \caption{{Illustration of various steps in the LFRE Algorithm. (a)~Node $i$ utilizes estimates from trusted neighbours in $\mathcal{N}_{i}^{(j)}$ to update its state. (b) If $\mathcal{N}_{i}^{(j)}$ has no trusted node but contains three distinct colored nodes (blue, red, and white in Fig.~\ref{fig:LFRE}(b)), then node $i$ sorts estimates from nodes in $\mathcal{N}_{i}^{(j)}$ in descending order. From the top (respectively bottom), node $i$ determines the smallest index $n_m$ (respectively largest index $n_M$) of the node with color different than the color of the node with the most  extreme estimate on that end. It then only considers estimates from nodes with indices within the range $n_m$ and $n_M$. (c) If $\mathcal{N}_{i}^{(j)}$ neither has a trusted node nor three distinct colored nodes, then node $i$ sorts the received estimates in descending order, removes the $f$-largest and $f$-smallest estimates, and considers only the remaining estimates in its update rule.}}
    \label{fig:LFRE}
\end{figure}
\begin{theorem}
Consider the system \eqref{eqn:plant} and measurement model \eqref{eqn:Obsmodel}, and suppose Assumption \ref{ass:system} holds.
 Let the communication graph $\mathcal{G}$ be strongly $((2f+1),\Delta(\cdot),\mathcal{T})$-robust w.r.t. $\mathcal{S}_j, \forall \lambda_j\in\Omega_{U}(\mathbf{A})$. Then, the LFRE algorithm for $f$-local mono-chromatic Byzantine adversaries solves Problem \ref{prob:resdist}.
\label{thm:main1}
\end{theorem}

\begin{proof}
Consider an $f$-local mono-chromatic Byzantine adversarial set $\mathcal{A}$, and let $\mathcal{R}=\mathcal{V}\setminus\mathcal{A}.$ 
Based on Lemma \ref{lemma:luen}, notice that a regular node $i\in\mathcal{R}$ can asymptotically estimate each component of the state vector $\mathbf{x}[k]$ corresponding to its set of detectable eigenvalues $\mathcal{O}_i$. It remains to show that node $i \in \mathcal{R}$ can also recover $x^{(j)}[k], \, \forall \lambda_j \in \bar{\mathcal{O}}_i$, based on the LFRE algorithm for $f$-local mono-chromatic adversaries. To this end, we argue that for each $\lambda_j \in \Omega_{U}(\mathbf{A})$, $\hat{x}^{(j)}_i[k]$ converges to $x[k]$ asymptotically for all $i\in\mathcal{R}$. 

Consider a mode $\lambda_j\in\Omega_{U}(\mathbf{A})$, and notice that based on Theorem \ref{thm:MEDAG_existence},  there exists a sub-graph $\mathcal{G}_j$ satisfying all the properties of a $(2f+1,\Delta(\cdot),\mathcal{T})$ MEDAG. Specifically, the set of regular nodes $\mathcal{R}=\mathcal{V}\setminus\mathcal{A}$ can be partitioned into disjoint levels $\{\mathcal{L}^{(j)}_0, \ldots, \mathcal{L}^{(j)}_q, \ldots, \mathcal{L}^{(j)}_{T_j}\}$. We induct on the level number $q$. For $q=0$, recall that $\mathcal{L}^{(j)}_0=\mathcal{S}_j\cap\mathcal{R}$ by definition. Hence, it follows from Lemma \ref{lemma:luen} that for each $i\in\mathcal{L}^{(j)}_0$, $\lim_{k\to\infty}e^{(j)}_i[k]=0$, where $e^{(j)}_i[k]=\hat{x}^{(j)}_i[k]-x^{(j)}[k]$. Next, consider a node $i$ in level $q=1$. We split our subsequent analysis into three separate cases.

\textbf{Case 1}: Suppose $\mathcal{N}^{(j)}_i\cap\mathcal{T}\neq\emptyset$. Then, based on Step 2 of the LFRE algorithm for $f$-local mono-chromatic adversaries, node $i$ employs the update rule \eqref{eqn:LFRE_1}. In this case, the error $e^{(j)}_i[k]$ evolves as follows:
\begin{equation}
{e}^{(j)}_i[k+1]=\lambda_j\left(\sum_{l \in \mathcal{N}^{(j)}_i\cap\mathcal{T}}\bar{w}^{(j)}_{il}{e}^{(j)}_{l}[k]\right),
\label{err1}
\end{equation}
where we used that (i) $x^{(j)}[k+1]=\lambda_jx^{(j)}[k]$ based on the structure of the $\mathbf{A}$ matrix, and (ii) the convexity of the weights $\bar{w}^{(j)}_{il}$. Based on the fact that $\mathcal{T}\subseteq\mathcal{R}$, and property (ii) of a MEDAG in Defn. \ref{defn:MEDAG_r,delta,T)}, we have that $\mathcal{N}^{(j)}_i \cap \mathcal{T} \subseteq \mathcal{L}^{(j)}_0$. It then follows from \eqref{err1} and the foregoing discussion that $\lim_{k\to\infty}e^{(j)}_i[k]=0$.

\textbf{Case 2}: Suppose $\mathcal{N}^{(j)}_i\cap\mathcal{T}=\emptyset$, but there exist three distinct colored nodes in $\mathcal{N}^{(j)}_i$. Then, based on Step 3 of the filtering algorithm, node $i$ employs the update rule \eqref{eqn:LFRE_2}.  In this case, the error $e^{(j)}_i[k]$ evolves as follows:
\begin{equation}
{e}^{(j)}_i[k+1]=\lambda_j\left(\sum_{l \in \mathcal{R}^{(j)}_i[k]}\tilde{w}^{(j)}_{il}[k]{e}^{(j)}_{l}[k]\right),
\label{eqn:err2}
\end{equation}
where we have once again used that $x^{(j)}[k+1]=\lambda_jx^{(j)}[k]$, and that the weights $\tilde{w}^{(j)}_{il}[k]$ are convex.
Observe that whenever $\mathcal{N}^{(j)}_i$ contains three distinct colored nodes, $\mathcal{R}^{(j)}_i[k]$ is guaranteed to be non-empty by definition. We now claim that at each time-step $k$, $e^{(j)}_l[k]$ lies in the convex hull of the points $e^{(j)}_s[k], s\in\mathcal{L}^{(j)}_0$, for all $l\in\mathcal{R}^{(j)}_i[k]$. To this end, fix a time-step $k$, and suppose that the node with the highest estimate of $x^{(j)}[k]$ in $\mathcal{N}^{(j)}_i$, namely node $n_1$, is regular. Then, we have that for each $l\in\mathcal{R}^{(j)}_i[k]$, $\hat{x}^{(j)}_l[k] \leq \hat{x}^{(j)}_{n_1}[k]$, where $n_1\in\mathcal{N}^{(j)}_i \cap \mathcal{R} \subseteq \mathcal{L}^{(j)}_0$. The last inclusion follows from property (ii) in Defn. \ref{defn:MEDAG_r,delta,T)}. Now consider the case when node $n_1$ is adversarial. Then, given the mono-chromaticity of the adversarial model, it must be that node $n_m$, as defined in Step 3, is regular, since $\Delta(n_m)\neq\Delta(n_1)$. Furthermore, based on how $\mathcal{R}^{(j)}_i[k]$ is defined in Step 3, it follows that for each $l\in\mathcal{R}^{(j)}_i[k]$, $\hat{x}^{(j)}_l[k] \leq \hat{x}^{(j)}_{n_m}[k]$, where $n_m\in\mathcal{N}^{(j)}_i \cap \mathcal{R} \subseteq \mathcal{L}^{(j)}_0$. Thus, we have established that at each time-step $k$, ${e}^{(j)}_l[k] \leq \max_{s\in\mathcal{L}^{(j)}_0} {e}^{(j)}_s[k]$, $\forall l\in\mathcal{R}^{(j)}_i[k].$ An identical argument reveals that at each time-step $k$, ${e}^{(j)}_l[k] \geq \min_{s\in\mathcal{L}^{(j)}_0} {e}^{(j)}_s[k]$, $\forall l\in\mathcal{R}^{(j)}_i[k].$ The above discussion, coupled with \eqref{eqn:err2}, and the fact that $\lim_{k\to\infty}e^{(j)}_s[k]=0, \forall s\in \mathcal{L}^{(j)}_0$, readily implies that $\lim_{k\to\infty}e^{(j)}_i[k]=0$.

\textbf{Case 3}: Suppose $\mathcal{N}^{(j)}_i\cap\mathcal{T}=\emptyset$, and there do not exist three distinct colored nodes in $\mathcal{N}^{(j)}_i$. Then, based on property (i) of a MEDAG in Defn. \ref{defn:MEDAG_r,delta,T)}, it must be that $|\mathcal{N}^{(j)}_i| \geq (2f+1)$.
In this case, node $i$ employs the update rule \eqref{eqn:LFRE_3}, which corresponds precisely to the resilient filtering algorithm developed in \cite{mitra_auto} for $f$-local Byzantine adversarial models. Thus, for this case, the fact that $\lim_{k\to\infty}e^{(j)}_i[k]=0$ follows directly from the arguments in \cite{mitra_auto}. 

This completes the analysis for the base case $q=1$. Using arguments similar to those for the base case, and a simple inductive reasoning as in \cite{mitra_auto}, one can establish that the result holds for all levels $q\in\{1,\ldots,T_j\}$. 
\end{proof}

\subsection{Resilient Distributed State Estimation Under Mono-chromatic Byzantine Adversaries}
\label{sec:mono_chrom_adv}
We now briefly discuss how the developments in the previous section can be easily generalized to account for a more powerful adversarial model wherein the assumption of $f$-locality is relaxed, i.e., we no longer require the adversarial set $\mathcal{A}$ to be $f$-local. We will, however, continue to assume that $\mathcal{A}$ is mono-chromatic and, to make the discussion meaningful, that $\mathcal{A}\subset\mathcal{V}.$ The appropriate concept that we need here is $(\infty,\Delta(\cdot),\mathcal{T})$-reachability, to be henceforth referred to as $(\Delta(\cdot),\mathcal{T})$ reachability - a special case of $(r,\Delta(\cdot),\mathcal{T})$-reachability in Defn. \ref{defn:r_delta_tau_reachable} with $r=\infty$, where the reachability condition can clearly only be satisfied via diversity or trust. The more stringent concept of  $(\Delta(\cdot),\mathcal{T})$-reachability seeks to make up for the inadequacy of the traditional notion of redundancy in coping with a mono-chromatic Byzantine adversarial model. Indeed, once $f$-locality is relaxed, a node may have direct or indirect paths from several informative nodes and, yet, fall short of estimating the state dynamics. In particular, an adversary can compromise all such informative nodes if they are of the same type, and not a part of the trusted set $\mathcal{T}$. This highlights the importance of incorporating diversity and/or trust into the measurement and communication structure of the network as alternatives to incorporating redundancy. 

Note that a strongly $(\Delta(\cdot),\mathcal{T})$-robust graph w.r.t. $\mathcal{S}$ and a $(\Delta(\cdot),  \mathcal{T})$ MEDAG are simply special cases of Defn.'s \ref{defn:strong_rdeltatau} and \ref{defn:MEDAG_r,delta,T)}, respectively, where the redundancy parameter is $\infty$. Then, following identical arguments as in Thm. \ref{thm:MEDAG_existence}, one can establish that for each $\lambda_j \in \Omega_{U}(\mathbf{A})$, $\mathcal{G}$ contains a  $(\Delta(\cdot),\mathcal{T})$ MEDAG for $\lambda_j$ if and only if $\mathcal{G}$ is strongly $(\Delta(\cdot),\mathcal{T})$-robust w.r.t. $\mathcal{S}_j$. To estimate its locally undetectable portion of the state, suppose each node $i\in\mathcal{R}$ executes only the first 3 steps of the filtering algorithm in Section \ref{subsec: floc_algo}, to update $\hat{x}^{(j)}_i[k], \forall \lambda_j\in\bar{\mathcal{O}}_i$. Let us call this algorithm the LFRE algorithm for mono-chromatic Byzantine adversaries. We then have the following result. 
\begin{theorem}
Consider the system \eqref{eqn:plant} and measurement model \eqref{eqn:Obsmodel}, and suppose Assumption \ref{ass:system} holds.
Let the communication graph $\mathcal{G}$ be strongly $(\Delta(\cdot),\mathcal{T})$-robust w.r.t. $\mathcal{S}_j, \forall \lambda_j\in\Omega_{U}(\mathbf{A})$. Then, the LFRE algorithm for mono-chromatic Byzantine adversaries solves the variant of Problem \ref{prob:resdist} corresponding to a mono-chromatic Byzantine adversary model.
\label{thm:main2}
\end{theorem}
\begin{proof}
The proof is similar to that of  Theorem \ref{thm:main1}. 
\end{proof}
\begin{remark}
\textbf{(Implications for Countering Spoofing Attacks)}:
Recently, in the context of multi-robot coordination, the authors in \cite{spoof1,spoof2} propose methods to tackle the so called ``Sybil attack", where an attacker spoofs or impersonates the identities of existing agents to gain a disproportionate advantage in the network. The methods developed in \cite{spoof1,spoof2} are based on analyzing the physics of wireless signals. Since such signals are invariably corrupted by environment and channel noise, the guarantees in \cite{spoof1,spoof2} are of a probabilistic nature. In contrast, we claim that the ideas developed in this section can provide deterministic guarantees in the face of spoofing attacks. The key enabling observation here is that even if an adversary generates multiple identities of an existing regular node, each such identity would share the same digital signature as that of the node being replicated. In other words, the node being spoofed along with its replicated identities would all be of the same type, or color. Thus, regardless of the number of fake identities, as long as the conditions in Theorem \ref{thm:main2} are met, our techniques would go through.
\label{rem:spoofing}
\end{remark}

\section{On the Complexity of Incorporating Diversity and Trust}
\label{sec:design}
 In practice, hardening sensors against attacks (i.e., making nodes trusted), and implementing several variants of nodes (i.e., making the network diverse), comes at a cost. Thus, it is natural to consider the design problem of (i) finding a trusted set of minimum cardinality; and/or (ii) finding the minimum number of colors, and the corresponding allocation of colors to nodes, so as to make the resulting network strongly-robust to a desired extent. In what follows, we separately explore the complexity of each of these problems.

\subsection{On the Complexity of Selecting Trusted Nodes}
To isolate the complexity of selecting trusted nodes, we consider a scenario where all nodes are of the same color (i.e., $\Delta(i)=\Delta(j), \forall i,j\in\mathcal{V}$). To proceed, we formally state the problem of interest and then characterize its complexity. 

\begin{problem} (\textbf{Trusted Strong-Robustness Augmentation (TSRA))} Given a system model \eqref{eqn:plant}, a measurement model \eqref{eqn:Obsmodel}, a communication graph $\mathcal{G}=(\mathcal{V},\mathcal{E})$ where all nodes are of the same color (i.e., $\Delta(i)=\Delta(j), \forall i,j\in\mathcal{V}$), and positive integers $r,t$, does there exist a set of trusted nodes $\mathcal{T}$ of cardinality $t$, such that $\mathcal{G}$ is strongly $(r,\Delta(\cdot),\mathcal{T})$-robust w.r.t. $\mathcal{S}_j$, $\forall \lambda_j \in \Omega_{U}(\mathbf{A})$?
\label{prob:TSRA}
\end{problem}  

To characterize the complexity of the TSRA problem, we will provide a reduction from the NP-hard Set Cover (SC) problem, defined as follows.

\begin{definition} (\textbf{Set Cover (SC)}) Given a collection of elements $\mathcal{U}=\{1, \ldots, p\}$, a set of subsets $\mathcal{F}=\{\mathcal{F}_1,\ldots,\mathcal{F}_m\}$ of $\mathcal{U}$, and a positive integer $t$, do there exist $t$ subsets in $\mathcal{F}$ whose union is $\mathcal{U}$?
\end{definition}

\begin{theorem} 
The TSRA problem is NP-complete.
\label{thm:TSRAP}
\end{theorem}
\begin{proof} 
We first argue that TSRA $\in$ NP. To see this, notice that for ``yes" instances of the problem, the set of trusted nodes $\mathcal{T}$ of size $t$ yields a certificate w.r.t. the MEDAG construction algorithm described in Section \ref{subsec:floc_graph}. Specifically, based on Theorem \ref{thm:MEDAG_existence}, for each $\lambda_j \in \Omega_{U}(\mathbf{A})$, the MEDAG construction algorithm terminates if and only if 
 $\mathcal{G}$ is strongly $(r,\Delta(\cdot),\mathcal{T})$-robust w.r.t. $\mathcal{S}_j$; thus, such an algorithm can be used to verify the desired graph property. That this verification algorithm has polynomial-time complexity follows from an analogous argument made in \cite[Proposition 2]{mitra_auto}.

Next, we establish that TSRA is NP-hard. To this end, given an instance of SC, we first construct an instance of TSRA as follows. We consider a scalar unstable dynamical system $x[k+1]=\lambda x[k]$, and construct an associated communication graph $\mathcal{G}$ with node set $\mathcal{V}=\bar{\mathcal{U}}\cup\bar{\mathcal{F}}$, where $\bar{\mathcal{U}}=\{u_1,\ldots,u_p\}$, and $\bar{\mathcal{F}}=\{f_1,\ldots,f_m\}$. For each $i\in\{1,\ldots,p\}$, node $u_i\in\bar{\mathcal{U}}$ corresponds to element $i$ of $\mathcal{U}$, and for each $j\in\{1,\ldots,m\}$, node $f_j\in\bar{\mathcal{F}}$ corresponds to subset $\mathcal{F}_j \in \mathcal{F}$. If $i \in \mathcal{F}_j$, then a directed edge is added from node $f_j$ to node $u_i$ in $\mathcal{G}$. Each node $f_j\in\bar{\mathcal{F}}$ is allocated a  non-zero measurement of the state $x[k]$. The cardinality of the trusted set $\mathcal{T}$ is set to $t$, and the desired level of strong-robustness is given by $r=|\mathcal{F}|$. Clearly, given any instance of SC, the above TSRA instance can be constructed in polynomial-time. We now argue that the answer to any given instance of SC is ``yes" if and only if the answer to the constructed instance of TSRA is ``yes".

Suppose the answer to the SC instance is ``yes". Thus, there exists a set of $t$ subsets of $\mathcal{F}$ whose union is $\mathcal{U}$. Without loss of generality, let these subsets be $\{\mathcal{F}_1, \ldots, \mathcal{F}_t\}$. Let the set of trusted nodes $\mathcal{T}$ be $\{f_1, \ldots, f_t\}$. We first observe that the set of source nodes $\mathcal{S}$ (the set of nodes that can detect $\lambda$)  of $\mathcal{G}$ is precisely the set $\bar{\mathcal{F}}$. Thus, $\mathcal{T}\subseteq\mathcal{S}$. To establish that $\mathcal{G}$ is strongly $(r,\Delta(\cdot),\mathcal{T})$-robust w.r.t. $\mathcal{S}$, we pick a non-empty subset $\mathcal{C}\subseteq\mathcal{V}\setminus\mathcal{S}=\bar{\mathcal{U}}$. Since $\{\mathcal{F}_1, \ldots, \mathcal{F}_t\}$ cover $\mathcal{U}$, $\mathcal{N}_{u_i}\cap\mathcal{T}\neq \emptyset, \forall u_i\in\bar{\mathcal{U}}$. Thus, $\mathcal{C}$ is $(r,\Delta(\cdot),\mathcal{T})$-reachable, and the answer to the constructed instance of TSRA is ``yes".

To show the converse, we proceed via contraposition. Suppose the answer to the SC instance is ``no". In other words, no $t$ subsets of $\mathcal{F}$ cover $\mathcal{U}$. Consider any set of trusted nodes $\mathcal{T}$ of cardinality $t$. Let $\mathcal{M}=\bar{\mathcal{F}}\cap\mathcal{T}$. We first consider the case when $\mathcal{M}$ is non-empty. In this case, there exists at least one node $u_i \in \bar{\mathcal{U}}$ that has neighbors (if any) only in $\bar{\mathcal{F}}\setminus\mathcal{M}$. Noting that the source set $\mathcal{S}=\bar{\mathcal{F}}$, we consider the non-empty set $\mathcal{C}=\{u_i\}$ contained in $\mathcal{V}\setminus\mathcal{S}$. Since $r=|\bar{\mathcal{F}}|$, it follows that $u_i$ neither has a trusted neighbor nor has at least $r$ neighbors. Thus, $\mathcal{C}$ is not $(r,\Delta(\cdot),\mathcal{T})$-reachable.\footnote{Note that as $\Delta(i)=\Delta(j),\forall i,j \in\mathcal{V}$ in TSRA, the requirements for $(r,\Delta(\cdot),\mathcal{T})$-reachability cannot be met via diversity (item (ii) in Defn. \ref{defn:r_delta_tau_reachable}).}  For analyzing the case when $\mathcal{M}$ is empty, we observe that there must exist at least one node $u_i\in\bar{\mathcal{U}}$ such that $\mathcal{N}_{u_i}\subset\bar{\mathcal{F}}$; else, each $\mathcal{F}_j\in\mathcal{F}$ would cover $\mathcal{U}$, and the answer to SC would be trivially ``yes", leading to a contradiction. It then follows that $\mathcal{C}=\{u_i\}$ is not $(r,\Delta(\cdot),\mathcal{T})$-reachable. Consequently, $\mathcal{G}$ is not strongly $(r,\Delta(\cdot),\mathcal{T})$-robust w.r.t. $\mathcal{S}$, regardless of the way $t$ trusted nodes are picked in $\mathcal{G}$. In other words, the answer to the constructed TSRA instance is ``no". This completes the proof.
\end{proof}

Given the above result, we now briefly describe a simple greedy heuristic that finds a potentially sub-optimal set of trusted nodes in polynomial time.

\textbf{Greedy Heuristic for Selecting Trusted Nodes}: Consider the setup in Problem \ref{prob:TSRA}, and suppose we need to find a set of trusted nodes $\mathcal{T}$ such that $\mathcal{G}$ is strongly $(r,\Delta(\cdot),\mathcal{T})$-robust w.r.t. $\mathcal{S}_j$, $\forall \lambda_j \in \Omega_{U}(\mathbf{A})$. We proceed as follows.  Fix a $\lambda_j\in \Omega_{U}(\mathbf{A})$, and suppose each node $i\in\mathcal{V}\setminus\mathcal{S}_j$ is reachable from $\mathcal{S}_j$ (since otherwise, there is no hope of achieving the desired property). Our proposed greedy algorithm proceeds in rounds $l$, where in each round precisely one node is made trusted, if needed. Two lists are maintained and updated each round: a list of ``active" nodes $\mathcal{W}_j(l)$, and a list of trusted nodes $\mathcal{T}_j(l)$, with $\mathcal{W}_j(0)$ initially set to $\mathcal{S}_j$, and $\mathcal{T}_j(0)$ to $\emptyset$. At the beginning of round $l$, where $l \geq 1$, each node in $\mathcal{W}_j(l-1)\setminus\mathcal{T}_j(l-1)$ is a candidate for being made trusted in that round. For each such candidate node $v\in\mathcal{W}_j(l-1)\setminus\mathcal{T}_j(l-1)$,  we run a virtual bootstrap percolation\footnote{Given a graph $\mathcal{G}$ and a threshold $r\geq2$, bootstrap percolation can be viewed as a process of spread of activation where one starts off with an initially active set. The process then evolves over the network in rounds, where in each round an inactive node becomes active if and only if it has at least $r$ active neighbors; here, we modify the activation rule to suit our purpose.} process by making node $v$ trusted temporarily, and computing the number of new nodes it activates in the process. Here, an inactive node gets activated if it either has at least $r$ active neighbors, or a trusted active neighbor. Let $\delta(v)$ denote the new nodes activated by node $v$. Having run this virtual percolation process separately for each $v\in\mathcal{W}_j(l-1)\setminus\mathcal{T}_j(l-1)$, we greedily pick $\tau(l)\in \argmax_{v\in\mathcal{W}_j(l-1)\setminus\mathcal{T}_j(l-1)} |\delta(v)|$ to be trusted in round $l$, i.e., we pick the node that activates the maximum number of nodes.  Subsequently, we update $\mathcal{W}_j(l)=\mathcal{W}_j(l-1)\cup\delta(\tau(l))$, and $\mathcal{T}_j(l)=\mathcal{T}_j(l-1)\cup\tau(l)$. Let $\bar{l}_j$ be the smallest integer such that $\mathcal{W}_j(\bar{l}_j)=\mathcal{V}$. We then say that the greedy algorithm described above terminates in round $\bar{l}_j$. It is easy to see that $\bar{l}_j \leq N-1$, and that on termination, $\mathcal{T}_j(\bar{l}_j)$ is such that $\mathcal{G}$ is strongly $(r,\Delta(\cdot),\mathcal{T}_j(\bar{l}_j))$-robust w.r.t. $\mathcal{S}_j$. Thus, we can run the above greedy heuristic for each $\lambda_j \in \Omega_{U}(\mathbf{A})$, and obtain the desired trusted set $\mathcal{T}=\cup_{\lambda_j \in \Omega_{U}(\mathbf{A})}\mathcal{T}_j(\bar{l}_j)$. 

A rigorous theoretical characterization of the performance of the above greedy heuristic is beyond the scope of this paper. However, it is not too hard to verify that this heuristic does output a trusted set of optimal size for simple graphs such as star graphs, directed trees, rings and complete graphs. 

\subsection{On the Complexity of Allocating Diversity}
We now turn our attention to the problem of allocating colors to the nodes from a set of specified cardinality so as to achieve a certain level of strong-robustness. To isolate the challenges associated with this problem, our subsequent analysis will focus exclusively on scenarios where the trusted set $\mathcal{T}$ is empty. Next, we formally state the problem of interest.

\begin{problem} (\textbf{$q$-Colored Strong-Robustness Augmentation ($q$-CSRA))} Given a system model \eqref{eqn:plant}, a measurement model \eqref{eqn:Obsmodel}, a communication graph $\mathcal{G}=(\mathcal{V},\mathcal{E})$ with an empty trusted set $\mathcal{T}$, and positive integers $r,q$, does there exist an allocation $\Delta:\mathcal{V}\rightarrow\{1,\ldots,q\}$, such that $\mathcal{G}$ is strongly $(r,\Delta(\cdot),\mathcal{T})$-robust w.r.t. $\mathcal{S}_j$, $\forall \lambda_j \in \Omega_{U}(\mathbf{A})$?
\end{problem}

Let us note that when $q < 3$, the $q$-CSRA problem as stated above boils down to checking whether the given graph $\mathcal{G}$ is strongly $r$-robust w.r.t. $\mathcal{S}_j$, $\forall \lambda_j \in \Omega_{U}(\mathbf{A})$. In \cite{mitra_auto}, by exploiting a connection to the process of bootstrap percolation, it was shown that this can be done in polynomial-time. Thus, the complexity of the $q$-CSRA problem remains to be characterized only when $q\geq3$. In the remainder of this section, we establish that the $3$-CSRA problem is computationally hard by providing a reduction from the NP-complete $3$-Disjoint Set Cover (3-DSC) problem, defined as follows \cite{cardei}.

\begin{definition} (\textbf{3-Disjoint Set Cover (3-DSC)}) Given a collection of elements $\mathcal{U}=\{1, \ldots, p\}$, and a set of subsets $\mathcal{F}=\{\mathcal{F}_1,\ldots,\mathcal{F}_m\}$ of $\mathcal{U}$, can $\mathcal{F}$ be partitioned into three disjoint collections of subsets, such that the union of the subsets within each such collection covers $\mathcal{U}$?
\end{definition}

\begin{theorem} 
The 3-CSRA problem is NP-complete.
\label{thm:3CSRAP}
\end{theorem}
\begin{proof} 
The fact that CSRA $\in$ NP follows an analogous argument as in Theorem \ref{thm:TSRAP}. In particular, given any ``yes" instance of the problem, the associated allocation $\Delta$ yields a certificate w.r.t. the MEDAG construction algorithm in Section \ref{subsec:floc_graph} that acts as a polynomial-time verifier. 

Given an instance of  3-DSC, we construct an instance of 3-CSRA in a manner identical to that in the proof of Theorem \ref{thm:TSRAP}, and adhere to the notation used in that proof. Note however that unlike TSRA, the cardinality $t$ of the trusted set $\mathcal{T}$ plays no role in 3-CSRA, and hence requires no specification while constructing the instance of 3-CSRA. It is easy to see that given any instance of 3-DSC, the above 3-CSRA instance can be constructed in polynomial-time. We now argue that the answer to any given instance of 3-DSC is ``yes" if and only if the answer to the constructed instance of 3-CSRA is ``yes". Throughout the proof, we will assume that $|\mathcal{F}|\geq 3$, as otherwise, the answer to 3-DSC is trivially ``no". 

Suppose the answer to the 3-DSC instance is ``yes". Thus, $\mathcal{F}$ can be partitioned into 3 disjoint set covers of $\mathcal{U}$. Let these partitions be denoted ${\mathcal{P}}_1=\{\mathcal{F}_{i_1},\ldots,\mathcal{F}_{i_{p_1}}\}$, ${\mathcal{P}}_2=\{\mathcal{F}_{j_1},\ldots,\mathcal{F}_{j_{p_2}}\}$, and ${\mathcal{P}}_3=\{\mathcal{F}_{k_1},\ldots,\mathcal{F}_{k_{p_3}}\}$, where $p_i=|\mathcal{P}_i|,i\in\{1,2,3\}.$ Let the corresponding sets of nodes in $\bar{\mathcal{F}}$ be denoted $\bar{\mathcal{P}}_1$, $\bar{\mathcal{P}}_2$ and $\bar{\mathcal{P}}_3$. Consider the following allocation of colors to the nodes in $\bar{\mathcal{F}}$ : $\Delta(f_{i_s})=1, \forall i_s\in\bar{\mathcal{P}}_1$, $\Delta(f_{j_s})=2, \forall j_s\in\bar{\mathcal{P}}_2$, and $\Delta(f_{k_s})=3, \forall k_s\in\bar{\mathcal{P}}_3$. The assignment of colors to the nodes in $\bar{\mathcal{U}}$ is arbitrary, i.e., each $u_i\in\bar{\mathcal{U}}$ is assigned any one of the three colors. Noting that the set of source nodes $\mathcal{S}$ is precisely the set $\bar{\mathcal{F}}$, we claim that $\mathcal{G}$ is strongly $(r,\Delta(\cdot),\mathcal{T})$-robust w.r.t. $\mathcal{S}$. To see this, pick any non-empty subset $\mathcal{C}\subseteq\mathcal{V}\setminus\mathcal{S}=\bar{\mathcal{U}}$. Since ${\mathcal{P}}_1$, ${\mathcal{P}}_2$ and ${\mathcal{P}}_3$ each cover $\mathcal{U}$, it follows that every $u_i\in\bar{\mathcal{U}}$ has a neighbor in each of the sets $\bar{\mathcal{P}}_1$, $\bar{\mathcal{P}}_2$ and $\bar{\mathcal{P}}_3$, i.e., each $u_i\in\bar{\mathcal{U}}$ has 3 distinct colored neighbors. Thus, $\mathcal{C}$ is $(r,\Delta(\cdot),\mathcal{T})$-reachable, and the answer to the constructed instance of 3-CSRA is ``yes". 

We now establish the converse. Suppose the answer to the 3-DSC instance is ``no". In other words, no matter how one partitions $\mathcal{F}$ into 3 disjoint collections of subsets, not all three such collections can each cover $\mathcal{U}$. We first argue that $\mathcal{G}$ cannot be made strongly $(r,\Delta(\cdot),\mathcal{T})$-robust w.r.t. $\mathcal{S}$, if one uses fewer than three colors to color the set $\bar{\mathcal{F}}$. To see this, note that if fewer than three colors are used to color $\bar{\mathcal{F}}$, then $\mathcal{G}$ will be strongly $(r,\Delta(\cdot),\mathcal{T})$-robust w.r.t. $\mathcal{S}$ if and only if each $f_j\in\bar{\mathcal{F}}$ is a neighbor of every $u_i\in\bar{\mathcal{U}}$, since each $u_i$ would need to have precisely $r=|\bar{\mathcal{F}}|$ neighbors to meet the $(r,\Delta(\cdot),\mathcal{T})$-reachability requirement (recall that $\mathcal{T}=\emptyset$). However, that would imply $\mathcal{F}_j=\mathcal{U}, \forall \mathcal{F}_j\in\mathcal{F}$. This in turn would collapse the size of the set $\mathcal{F}$ to just 1 (since all its elements would be identical), contradicting the fact that $|\mathcal{F}| \geq 3$. 

Next, consider any allocation of these colors to the nodes in $\bar{\mathcal{F}}$, where each of the three colors is used at least once. Such a coloring naturally partitions $\bar{\mathcal{F}}$ into 3 disjoint non-empty sets, say $\bar{\mathcal{P}}_1$, $\bar{\mathcal{P}}_2$ and $\bar{\mathcal{P}}_3$. Since the answer to 3-DSC is ``no", there must exist some node $u_i\in\bar{\mathcal{U}}=\mathcal{V}\setminus\mathcal{S}$, such that $u_i$ contains at most 2 distinct colored neighbors from $\bar{\mathcal{F}}$. Since $\bar{\mathcal{P}}_1$, $\bar{\mathcal{P}}_2$ and $\bar{\mathcal{P}}_3$ are each non-empty, and $r=|\bar{\mathcal{F}}|=|\bar{\mathcal{P}}_1|+|\bar{\mathcal{P}}_2|+|\bar{\mathcal{P}}_3|$, it follows that $|\mathcal{N}_{u_i}| < r$. Consequently, $\{u_i\}\in\mathcal{V}\setminus\mathcal{S}$ is not $(r,\Delta(\cdot),\mathcal{T})$-reachable. Based on the above discussion, we conclude that there does not exist any allocation $\Delta:\mathcal{V}\rightarrow\{1,2,3\}$ that renders $\mathcal{G}$ strongly $(r,\Delta(\cdot),\mathcal{T})$-robust w.r.t. $\mathcal{S}$. The answer to the constructed instance of $3$-CSRA is thus ``no". This completes the proof.
\end{proof}

At the moment, we do not have a clean heuristic algorithm to allocate diversity; we reserve this as future work. 
\section{Conclusion}
We introduced novel graph-theoretic constructs to study the impacts of redundancy, diversity, and trust in the context of resilient distributed state estimation. We then proposed an attack-resilient algorithm that appropriately leverages each of the three above facets, and provides provable guarantees. Roughly speaking, we established that even relatively sparse networks that are either diverse, or contain a small subset of trusted nodes, can exhibit the same functional robustness as densely connected networks. Finally, we separately studied the complexity of (i) selecting a trusted node set, and (ii) allocating diversity, in order to achieve a prescribed level of robustness. Our analysis revealed that each such problem is NP-complete; in the future, we plan to explore approximation algorithms with provable guarantees for each of these problems.
\bibliographystyle{IEEEtran} 
\bibliography{refs}
\end{document}